%% file: arbm.tex
\documentclass{llncs}

\title{Axiomatizing Resource Bounds for Measure \thanks{A preliminary
    version of the paper was presented in the Workshop on Logic in
    Computational Complexity, 2009.}}

\titlerunning{Axiomatizing Resource Bounds for Measure}
\authorrunning{Gu, Lutz, Nandakumar, Royer}   

\author{
  Xiaoyang Gu
  \inst{1}
  \thanks{This author's research was supported in part by National
    Science Foundation Grants 0344187, 0652569 and 0728806.} 
  \and 
  Jack H. Lutz
  \inst{2}
  \thanks{This author's research was supported in part by
    Spanish Government MEC Project TIN 2005-08832-C03-02.}
  \and
  Satyadev Nandakumar
  \inst{3} 
  \and James S. Royer
  \inst{4}
  \thanks{This author's research was supported in part by National
  Science Foundation Grant CCR-0098198.}}

\institute{
  LinkedIn Corporation,\\
  Mountain View, CA 94043, U. S. A.\\
  \email{xgu@linkedin.com}
  \and
  Iowa State University\\
  Ames, IA 50011, U. S. A.\\
  \email{lutz@cs.iastate.edu}
  \and
  I. I. T. Kanpur\\
  Kanpur, UP 280 016 India.\\
  \email{satyadev@cse.iitk.ac.in}
  \and
  Syracuse University\\
  Syracuse,  New York  13244,  U. S. A.\\
  \email{royer@ecs.syr.edu}}

\date{} 

\usepackage{amsfonts}
\usepackage{amsmath}
\usepackage{amssymb}

\usepackage[dvips]{graphics}

\input{commands.tex}


\renewcommand{\include}{\input}
\newcommand{\func}[3]{ {#1} : {#2} \rightarrow {#3} }

\begin{document}

\maketitle

\newcommand{\BFF}{\mathrm{BFF}}
\newcommand{\BFSF}{\mathrm{BFSF}}
\newcommand{\Tone}{\mathcal{T}_1}
\newcommand{\Ttwo}{\mathcal{T}_2}
\newcommand{\Ap}{\mathbf{Ap}}
\newcommand{\pspaceII}{\pspace_\mathrm{II}}

\begin{abstract}
Resource-bounded measure is a generalization of classical
Lebesgue measure that is useful in computational complexity.
The central parameter of resource-bounded measure is the
{\it resource bound} $\Delta$, which is a class of functions.
When $\Delta$ is unrestricted, i.e., contains all functions
with the specified domains and codomains, resource-bounded
measure coincides with classical Lebesgue measure.  On the
other hand, when $\Delta$ contains functions satisfying some
complexity constraint, resource-bounded measure imposes
internal measure structure on a corresponding complexity
class.

Most applications of resource-bounded measure use only the
``measure-zero/measure-one fragment'' of the theory.  For this
fragment, $\Delta$ can be taken to be a class of type-one functions
(e.g., from strings to rationals).  However, in the full theory of
resource-bounded measurability and measure, the resource bound
$\Delta$ also contains type-two functionals.  To date, both the full
theory and its zero-one fragment have been developed in terms of a
list of example resource bounds chosen for their apparent utility.

This paper replaces this list-of-examples approach with a
careful investigation of the conditions that suffice for
a class $\Delta$ to be a resource bound.  Our main theorem
says that every class $\Delta$ that has the closure properties
of Mehlhorn's basic feasible functionals is a resource bound
for measure.

We also prove that the type-2 versions of the time and space
hierarchies that have been extensively used in resource-bounded
measure have these closure properties. In the course of doing this, we
prove theorems establishing that these time and space resource bounds
are all robust.
\end{abstract}
\include{section1}
\include{section2}

\include{section3}
\include{section4}
\include{section5}
%
%
\include{bib}

\newpage


\end{document}

%% file: commands.tex
\usepackage{amsfonts}
\usepackage{amsmath}
\usepackage{amssymb}
\usepackage{enumerate}

\newcommand{\PROOF}{\begin{proof}}

\newcommand{\QED}{\end{proof}}

\newcommand{\myset}[2]{ \left\{ #1 \;\left|\; #2 \right. \right\} }

\newcommand{\prefix}{\sqsubseteq}
\newcommand{\prprefix}{\underset{\neq}{\sqsubset}}

\newcommand{\liminfn}{\liminf\limits_{n\to\infty}}
\newcommand{\limsupn}{\limsup\limits_{n\to\infty}}

\newcommand{\N}{\mathbb{N}}
\newcommand{\Z}{\mathbb{Z}}
\newcommand{\Q}{\mathbb{Q}}
\newcommand{\R}{\mathbb{R}}

\newcommand{\p}{{\mathrm{p}}}

\newcommand{\pspace}{{\mathrm{pspace}}}
\newcommand{\ptwospace}{{\mathrm{p}_{\thinspace\negthinspace_2}\mathrm{space}}}
\newcommand{\pispace}{{\mathrm{p}_{\thinspace\negthinspace_i}\mathrm{space}}}

\newcommand{\str}{{\mathrm{str}}}

\newcommand{\D}{{\mathcal{D}}}

\newcommand{\C}{\mathbf{C}}

\newcommand{\pred}{\mathrm{pred}}

\newcommand{\strings}{\{0,1\}^*}

\newcommand{\DTIME}{\mathrm{DTIME}}

\newcommand{\NP}{{\ensuremath{\mathrm{NP}}}}

\newcommand{\E}{{\rm E}}

\renewcommand{\P}{\ensuremath{{\mathrm P}}}

\newcommand{\F}{{\mathcal{F}}}



%% file: section1.tex
\section{Introduction}
Resource-bounded measure is a generalization of classical Lebesgue
measure theory that allows us to quantify the ``sizes'' (measures) of
interesting subsets of various complexity classes. This quantitative
capability has been useful in computational complexity because it has
intersected informatively with reducibilities, completeness,
randomization, circuit-size, and many other central ideas of
complexity theory. Resource-bounded measure has given us a
generalization of the probabilistic method that works inside
complexity classes (leading, for example, to improved lower bounds on
Boolean circuit size \cite{Lutz:AEHNC} and the densities of complete
problems \cite{Lutz:MSDHL}) and new complexity-theoretic hypotheses
(e.g., the hypothesis that $\NP$ is a non-measure $0$ subset of
exponential time) with many plausible consequences, i.e., significant
explanatory power. The somewhat outdated survey papers
\cite{BuhTor94,AmbMay97,AmbMay97,BuhTor98,Lutz:TPRBM,Pavan03} and more
recent papers in the bibliography \cite{RBM-bib} give a more detailed
account of the scope of resource-bounded measure and its applications.

The central parameter in resource-bounded measure is the {\em resource
  bound}, which is a class $\Delta$ of functions. When $\Delta$ is
unrestricted, i.e., contains {\em all} functions with specified
domains and codomains, resource-bounded measure coincides with
classical Lebesgue measure on the Cantor space $\C$ of all decision
problems.  On the other hand, when $\Delta$ only contains functions
satisfying a suitable complexity constraint, resource-bounded measure
consists of the following two theories.
\begin{enumerate}[1.]
\item
A theory of $\Delta$-{\em measure}. This is a
``$\Delta$-constructive'' measure theory on $\C$.
\item
A theory of {\em measure in} a complexity class $R(\Delta)$. This is a
theory that $\Delta$-measure imposes on the ``result class''
$R(\Delta)$.
\end{enumerate}
(Result classes and other notions discussed informally in this
introduction are defined precisely in the sections that follow.)  For
example, if $\Delta=\p$ consists of functions that are computable in
polynomial time, then we have $\p$-measure on $\C$, and this imposes
an internal measure structure on the exponential time complexity class
$R(\p) =\E =\DTIME(2^{\mathrm{linear}})$.  Typically, one proves a
result on measure in $R(\Delta)$ by proving a corresponding result on
$\Delta$-measure.  This, together with the fact that the
$\Delta$-measure result implies a corresponding $\Delta'$-measure
result for every resource bound $\Delta'\supseteq \Delta$, provides
resource-bounded measure a substantial underlying unity.

Of the hundred or so papers that have been written about
resource-bounded measure since 1992, none gives a definition of the
term ``resource bound''. Most simply work with those few resource
bounds appropriate to the complexity-theoretic problems being
investigated. Even papers of a more general nature stipulate that the
resource bound $\Delta$ is one of a specified (infinite) list of
examples chosen for their prior utility.

This approach to resource bounds has been healthy for the initial
development of a theory intended as a tool, but, as Socrates taught us
in {\em Euthyphro} \cite{Plato}, a list of examples leaves us far
short of understanding a concept.  More pragmatically, as the list
grows, it becomes ever more burdensome to verify that a theorem about
a general resource bound $\Delta$ actually holds for all examples in
the list.

This paper shows that there is a simple and natural set of axioms with
the following two properties.

\begin{itemize}
\item Adequacy: Any class $\Delta$ satisfying the axioms can be used
  as a resource bound for measure.
\item Generality: The most extensively used resource bounds satisfy
  the axioms.
\end{itemize}

We thus propose to \emph{define} a resource bound to be a class
$\Delta$ satisfying the axioms.

What makes our task challenging is the fact that, in order to define
resource-bounded measurability and measure \cite{Lutz:RBM} a resource
bound $\Delta$ must contain not only functions on discrete domains
like $\strings$ and $\N$, but also type-2 \emph{functionals} that take
functions as arguments. It has been a major undertaking to define what
it means for such functionals to be feasible (computable in polynomial
time) and to verify that the definition is robust
\cite{IKR01,DR:ATS:LMCS}.  The second author \cite{Lutz:RBM} has
defined type-2 versions of the other time and space resource bounds
that have been extensively used in resource-bounded measure (the
quasi-polynomial time and space hierarchies). However, these
definitions have not been proven to be robust, and the machine-based
definitions of \cite{Lutz:RBM}, while proven to be sufficient for the
development of measure and measurability, shed very little light on
our present question, namely, what properties of a class of type-2
functionals make it an adequate resource bound for measure.

Fortunately, it turns out that an {\em existing} set of axioms can be
adapted to our purpose. Mehlhorn's {\em basic feasible functionals}
\cite{Mehl76} were originally defined as a function algebra, i.e., a
set of initial functionals and a set of closure properties, with the
understanding that the class of basic feasible functionals is the {\em
  smallest} class containing these initial functions and enjoying
these closure properties.

The main contribution of the present paper is to demonstrate that, if
we just discard the ``smallest'' proviso in Mehlhorn's scheme and
define a resource bound to be {\em any} class of functionals
containing the initial functions and having the closure properties of
his definition, then we will, indeed have a definition that is
sufficient for the development of measurability and measure in
\cite{Lutz:RBM}.

We also prove that all the classes in the quasi-polynomial time and
space hierarchies of \cite{Lutz:RBM} are resource bounds in this
sense.  In the course of proving this, we prove new function algebra
characterizations of these classes, thereby establishing that they are
robust.

Two additional remarks on related work are in order here.  First,
there has been work on resource-bounded measure that is not captured
by our axiomatization.  The notable examples here are the measures in
``small'' complexity classes (e.g., the polynomial time class P)
developed by Moser \cite{Moser08} (building on pioneering work of
Mayordomo \cite{Mayo94a} and Allender and Strauss \cite{AllStr94}),
the measures in probabilistic classes (e.g., the randomized
exponential time class BPE) developed by Moser \cite{Moser08b}, and
the measures in ``large'' complexity classes (e.g., the doubly
exponential time class EE) developed by Harkins and Hitchcock
\cite{Harkins:UM}.  To date, this work has all been confined to
measure $0$/measure $1$ results.  Future developments of general
measurability and measure in these settings may necessitate -- and
guide -- generalizations of the axiomatization presented here. This
remains an open question.

The other line of related work that we mention is Dai's outer measure
approach to measurability and measure in complexity classes
\cite{Dai09}.  This approach is simpler than that of \cite{Lutz:RBM}
and the present paper in that it does not require type-two
functionals.  On the other hand, the approach of \cite{Dai09} only
seems to yield theory 2 in the second paragraph of this introduction,
so that all results are ``local'' to a particular complexity class.
The unity provided by theory 1 above, i.e., a ``global''
$\Delta$-measure on all of Cantor space, is a substantial advantage of
our our present approach.  Only future research will determine whether
a single approach can achieve both the simplicity of \cite{Dai09} and
the unity of \cite{Lutz:RBM}.

The rest of the paper is organized as follows. Section \ref{se:2}
gives the preliminary definitions and notational conventions. Section
\ref{se:3} describes the classes of type-two functionals used in the
paper.  Section \ref{se:4} gives the definition of a resource bound
and shows that the standard resource bounds in the literature satisfy
the definition. Section \ref{se:5} establishes that the definition of
a resource bound is adequate to establish the fundamental theorems of
resource-bounded measure. The final section proves that the
measure-zero fragment of this theory coincides with the approach
current in the literature \cite{Lutz:RBM}.

%% file: section2.tex
\section{Preliminaries}\label{se:2}
\newcommand{\bton}{\mathrm{bton}} 
\newcommand{\ntob}{\mathrm{ntob}}

We use a binary alphabet $\{0,1\}$ in this paper.  A string is an
element in $\strings$. For every $w\in\strings$, $|w|$ is the length
of the string $w$, and $w[i]$ denotes the $i$th bit of $w$.  The {\em
  Cantor} space $\C = \{0,1\}^\infty$ is the set of all infinite
binary sequences.  For an $S\in \C$, $S[i]$ is the $i$th bit of $S$,
and $S[0..n-1]$ is the $n$-bit prefix of $S$.  The standard
enumeration of $\strings$ is the enumeration of all strings in
$\strings$ in increasing order of length, with strings of the same
length ordered lexicographically.  The binary encoding function is
$\ntob:\N\rightarrow\strings$ such that for all $n\in\N$, $\ntob(n)$
is the $n$th string in the standard enumeration.  The binary decoding
function $\bton:\strings\rightarrow\N$ is the inverse of the binary
encoding function. For example, $\bton(\lambda) = 0$ and $\bton(01) =
4$.

The binary notational successor functions are
$s_0,s_1:\strings\rightarrow\strings$ such that $s_0(u) = u0$ and
$s_1(u)=u1$ for all $u\in\strings$.  The binary successor function is
$s:\strings\rightarrow\strings$ such that for all $u\in\strings$,
$s(u) = \ntob(\bton(u) +1))$ - that is, if $u$ represents a number
$n$, then $s(u)$ is the encoding of $n+1$.  The binary predecessor
function is $\pred:\strings\rightarrow\strings$ such that $\pred(u) =
\ntob(\max\{\bton(u)-1, 0\})$.

The smash function is $\#:\strings \times \strings \rightarrow
\strings$ such that for all $u,v\in\strings$, $\#(u,v) =
1^{|u|\cdot|v|}$. The interesting property of the smash function is
that for every pair $(u,v)$, the string $\#(u,v)$ has length equal to
the product of the lengths of $u$ and $v$.

A language is a subset of $\strings$. The characteristic sequence of
$L$ the infinite binary sequence such that $S[i] = 1\iff \ntob(i)\in
L$. Analogously, the characteristic function of $L$ is
$\chi_L:\strings\rightarrow\{0,1\}$ such that $\chi_L(x)=1\iff x\in
L$.  When no ambiguity arises, we also use $L$ for the characteristic
sequence of $L$.  

We write $w\prefix A$ if string $w$ is a prefix of a string/sequence
$A$.  A cylinder in $\C$ is a subset, of the form
$\myset{S\in\C}{w\prefix S}$ for some $w$, denoted $\C_w$. An {\em
  open set} in $\C$ is a set of the form $\bigcup_{w\in A} \C_w$ for
some $A\subseteq \strings$.

We also define the following hierarchy of functions.  Let $g_0=2n$ and
let $g_i(n) = 2^{g_{i-1}(\log n)}$ for all $i\in\Z^+$.  Note that
$g_1(n) = n^2$ and that $g_2(n) = n^{\log n}$.  For $i\in \N$, let
$G_i$ be the class of functions that contains $g_i$ and is closed
under composition. We use $G_i$ to represent different growth
rates. $G_1$ represents polynomial growth rates ($O(n^c)$) and $G_2$
represents quasi-polynomial growth rates ($O(n^{\log^cn})$). For each
$i\in\N$, we call growth rates bounded by a function in $G_i$ as
quasi$^i$-polynomials.

%% file: section3.tex
\section{Type--2 Functionals}\label{se:3}
In 1965, Cobham characterized type-1 polynomial-time computable
functions using limited/bounded recursion on notation
\cite{Cobh65,Weih73}.  He proved that the class of polynomial-time
computable functions is the smallest class of functions containing the
constant $0$ function, the binary notational successor functions, and
the smash function that is closed under composition and limited
recursion on notation.

Mehlhorn extended the characterization of polynomial-time
computability to type-2 functionals.

\begin{definition}[Mehlhorn \cite{Mehl76}]
$F$ is defined from $G$, $H$, $K$ by {\em limited recursion on
    notation} if for all $\vec f$, $\vec x$, $w$,
\[
\begin{array}[t]{ll}
F(\vec f, \vec x, \lambda) = G(\vec f, \vec x) \\ F(\vec f, \vec x,
wb) = H(\vec f, \vec x, wb, F(\vec f,\vec x, w)),& b \in \{0,1\}\\ |
F(\vec f, \vec x, w) | \leq |K(\vec f, \vec x, w)|.
 \end{array}\]
\end{definition}
We also use the following definition from Kapron and Cook
\cite{KapCoo96}.
\begin{definition}[Kapron and Cook \cite{KapCoo96}]
$F$ is defined from $H$, $G_1$, ..., $G_l$ by {\em functional
    composition} if for all $\vec{f}$, $\vec{x}$,
\[F(\vec{f},\vec{x})= H(\vec{f}, G_1(\vec{f}, \vec{x}), \dots, G_l (\vec{f},\vec{x})).\]
$F$ is defined from $G$ by {\em expansion} if for all $\vec{f}$,
$\vec{g}$, $\vec{x}$, $\vec{y}$,
\[F(\vec f, \vec g, \vec x, \vec y)= G(\vec f, \vec x).\]
\end{definition}
For the definition of basic feasible functionals, we adopt Kapron and
Cook's definition.
\begin{definition}[Kapron and Cook \cite{KapCoo96}]
Let $X$ be a set of type-two functionals. The class of {\em basic
  feasible functionals defined from $X$} ($\BFF(X)$) is the smallest
class of functionals that contains $X$, all polynomial-time functions
of type-one and the application functional $\Ap$, defined by $\Ap(f,
x) = f(x)$, and is closed under functional composition, expansion, and
limited recursion on notation.  The {\em basic feasible functionals}
are $\BFF(\varnothing)$.
\end{definition}

\begin{remark}
In this definition, it is possible to replace the inclusion of
polynomial-time functions in $\Tone$ with the inclusion of the
constant $0$ function, the binary notational successor functions, and
the smash function. But since Cobham's functional algebraic
characterization of polynomial-time is well-understood now, directly
using polynomial-time functions allow us to avoid repeating the
tedious process of defining all the simple functions from scratch.
\end{remark}

Mehlhorn proved that the $\BFF$'s have the {\em Ritchie-Cobham
  property}, namely, $F\in\BFF$ if and only if there exists an oracle
Turing machine $M$ and $G\in\BFF$ such that for all input $f$ and $x$,
the running time of $M(f,x)$ is bounded by $|G(f,x)|$. Mehlhorn's
result serves as partial evidence that the functional algebraic notion
of $\BFF$ is robust. On top of this, Kapron and Cook defined a notion
of type-2 polynomial-time computability based on oracle Turing
machines that does not require the use of $\BFF$ time-bound like the
one in Mehlhorn's result. The Basic Feasible Functionals capture the
notion of an intuitively feasible class of type-two functionals.
Basic Feasible Functionals and their probabilistic versions show up in
cryptography, for instance, in many constructions of pseudo-random
generators from one way functions. In \cite{IK06} the authors remark
that many cryptographic adversaries can be formalized as type-2
probabilistic feasible functionals or circuits.

First, we generalize Kapron and Cook's definition of second-order
polynomials to the following.
\begin{definition}
Let $i\in \Z^+$.  {\em First-order variables} are elements of the set
$\{ n_1, n_2, \dots\}$.  {\em Second-order variables} are elements of
the set $\{L_1, L_2,\dots\}$.  {\em Second-order
  quasi$^i$-polynomials} are defined inductively: any $c\in\N$ is a
second-order quasi$^i$-polynomial; first-order variables are
second-order quasi$^i$-polynomials; and if $P$, $Q$ are second-order
quasi$^i$-polynomials and $L$ is a second-order variable, then $P+Q$,
$P\cdot Q$, $L(P)$, and $g_i( P)$ are second-order
quasi$^i$-polynomials.
\end{definition}
Second-order quasi$^1$-polynomials are the second-order polynomials
defined by Kapron and Cook.  Second-order quasi$^2$-polynomials are
second-order quasi-polynomials.  They also defined a notion of the
length for type-1 functions.
\begin{definition}[Kapron and Cook\cite{KapCoo96}]
For any $f:\strings\rightarrow \strings$, the {\em length } of $f$ is
the function $|f|: \N\rightarrow \N$ defined by
\[|f|(n) = \max_{|w|\leq n} |f(w)|.\]
\end{definition}
Note that $|f|$ is non-decreasing.

With the above two definitions, Kapron and Cook defined the following
notion of polynomial-time bounded oracle Turing machine computation.
\begin{definition}
A type-two functional $F$ is basic poly time if there is an oracle
Turing machine $M$ and a second-order polynomial $P$ such that $M$
computes $F$, and for all $\vec f$ and $\vec x$, the running time of
$M(\vec f,\vec x)$ is bounded by $P(|f_1|, \dots, |f_k|, |x_1|,\dots,
|x_l|)$.
\end{definition}
Strongly confirming the robustness of the notion of $\BFF$s, they
proved the following.
\begin{theorem}[Kapron and Cook \cite{KapCoo96}]\label{th:kapcoo}
A functional $F$ is $\BFF$ if and only if it is basic poly time.
\end{theorem}

In this paper, we extend the Mehlhorn's functional algebraic notion of
feasible functionals to quasi-feasible functionals with the following
definition.

\begin{definition}
Let $X\subseteq \Ttwo$ and let $i\in \Z^+$. The class of {\em basic
  $i$-feasible functionals defined from $X$} ($\BFF_i(X)$) is the
smallest class of functionals containing $X$, all polynomial-time
computable functions in $\Tone$, $1^{g_i(|x|)}$, and the application
functional $\Ap$, defined by $\Ap(f, x) = f(x)$, and which is closed
under functional composition, expansion, and limited recursion on
notation.  The {\em basic $i$-feasible functionals} are elements of
the class $\BFF_i(\varnothing)$.
\end{definition}
In the flavor of Kapron and Cook, we extend their oracle Turing
machine based notion of feasible computation to the following.
\begin{definition}
Let $i\in \Z^+$.  A functional $F$ is basic quasi$^i$-polynomial time
if there is an oracle Turing machine $M$ and a second-order
quasi$^i$-polynomial $P$ such that $M$ computes $F$, and for all $\vec
f$ and $\vec x$, the running time of $M(\vec f,\vec x)$ is bounded by
$P(|f_1|, \dots, |f_k|, |x_1|,\dots, |x_l|)$.
\end{definition}
The following theorem is a corollary of Kapron and Cook's proof of
theorem \ref{th:kapcoo}.
\begin{theorem}
Let $i\in\Z^+$.  A functional $F$ is $\BFF_i$ if and only if it is
basic quasi$^i$-polynomial time.
\end{theorem}

In the machine model, the time bound is based on both the input length
and on the length of query answers.  This is why we need to have
$g_i(P)$ in the definition of second-order quasi$^i$-polynomials.  The
condition in the definition that a single second-order
quasi$^i$-polynomial has to work for all input $\vec f$ prohibits an
oracle Turing machine from using extra running time when the input
function $\vec f$ is pathologically long. An oracle Turing machine $M$
that computes a quasi$^i$-polynomial time functional, on any $x$, can
only utilize an amount of time that is quasi$^i$-polynomial in the
length of $\vec f$ it can provide evidence for, which can be much less
than length of $\vec f$ depending on the type-$0$ inputs. 

More formally, let $Q_x$ be the set of all queries made by $M$ with
$\vec f$ and $x$ as input.  Let $P$ be the time bound of $M$.  Let
$f_{Q_x}(y)=\vec f(y)$ if $y\in Q_x$ and $0$ otherwise.  Then the
running time $T_M(\vec f, x) \leq P(|f_{Q_x}|, |x|)$ for all $\vec f$
and $x$.  The key idea behind Kapron and Cook's proof is that it is
possible to find the oracle query $q_\mathrm{max}$ made by $M(\vec
f,x)$ that maximizes $|f(q_\mathrm{max})|$ in $\BFF$. And the
inability to compute the length of $\vec f$ (in unary) in $\BFF$ is
what makes their proof very involved.  We will see in the following
that the situation with polynomial space-bounded computation is much
simpler precisely for the reason that, as we will soon prove in Lemma
\ref{lm:lengthspace}, the length functional in unary for arbitrary
$\vec f$ is actually computable in polynomial space.  First, we
develop the definitions of computation feasible in terms of space.

\begin{definition}
A functional $F$ is {\em quasi$^i$-polynomial space} if there is an
oracle Turing machine $M$ and a second order quasi$^i$-polynomial $P$
such that $M$ computes $F$, and for all $\vec f$, $\vec x$, $S_M(\vec
f, \vec x)$ is bounded by $P(|f_1|, \dots, |f_k|, |x_1|, \dots,
|x_l|)$, where $S_M(\vec f, \vec x)$ is the running space used by $M$
on input $\vec f$ and $\vec x$.
\end{definition}

In 1972, D. B. Thompson characterized the class of type-1
polynomial-space computable functions as the smallest class that
contains the constant $0$ function, the binary successor function, the
smash function, and is closed under (type-1) composition and (type-1)
bounded recursion \cite{Thom72}.  We extend type-1 bounded recursion
as follows.
\begin{definition}
$F$ is defined from $G$, $H$, $K$ by {\em bounded recursion} (BR) if
  for all $\vec f$, $\vec x$, $n$,
\[
\begin{array}[t]{ll}
F(\vec f, \vec x, 0) = G(\vec f, \vec x) \\ F(\vec f, \vec x, n+1) =
H(\vec f, \vec x, n, F(\vec f,\vec x, n))\\ F(\vec f, \vec x, n)\leq
K(\vec f, \vec x, n).
 \end{array}\]
\end{definition}

\begin{definition}
Let $X\subseteq \Ttwo$ and let $i\in \Z^+$.  The class of {\em basic
  $i$-feasible space functionals defined from $X$} ($\BFSF_i(X)$) is
the smallest class of functionals containing $X$, all polynomial-time
computable functions in $\Tone$, $1^{g_i(|x|)}$, and the application
functional $\Ap$, defined by $\Ap(f, x) = f(x)$, and which is closed
under functional composition, expansion, and bounded recursion.  The
{\em basic $i$-feasible space functionals} are $\BFSF_i(\varnothing)$.
\end{definition}
\begin{lemma}\label{lm:lengthspace}
$L:(f,x)\mapsto 1^{|f|(|x|)}$ is basic $i$-feasible space for all
  $i\geq 1$.
\end{lemma}
\begin{proof}[Proof of Lemma \ref{lm:lengthspace}]
Let the functional
$$F:(\strings\to\strings)\times \strings\rightarrow \strings$$ be
defined using bounded recursion as follows:
\begin{align*}
F(f, \lambda) &= \lambda\\ F(f, x) &=\begin{cases} F(f,
\pred(x))&|f(F(f, \pred(x)))| \geq |f(x)|\\ x&\text{otherwise}
\end{cases}\\
F(f, x)&\leq x.\\
\end{align*}
Intuitively,
\[F(f,x) = \min\myset{y}{y\leq x\text{ and }|f(y)| = \max_{z\leq x}|f(z)|}.\]
Let
\[L(f,x) = 1^{|\Ap(f(F(f,1^{|x|})))|}.\]
Then $L$ is the functional we desire here.  \qed\end{proof}

\begin{theorem}\label{th:3_4}
A functional $F$ is basic $i$-feasible space if and only if it is
quasi$^i$-polynomial space.
\end{theorem}
\begin{proof}[Proof Sketch of Theorem \ref{th:3_4}.]
To prove that basic $i$-space feasibility implies quasi$^i$-polynomial
space, it suffices to do an induction on the structure of composition
and bounded recursion by implementing them on Turing machine with
space reuse.

For the other side of the equivalence, we prove this by an induction
on the depth of second-order polynomials.

Let $F$ be a quasi$^i$-polynomial space computable functional computed
by OTM $M$ with space bound $P$ of depth $d>0$. (When $d=0$, the
running space bound of $F$ does not depend on the length of the input
type-$1$ function $f$ and the proof is simpler.)  Then there exist
(regular) level-$i$ polynomials $q_0$, $q_1$, $\dots$, $q_d$, and
second-order quasi$^i$-polynomials $P_1$, $\dots$, $P_d$ such that for
all $i\in [1..d-1]$
\[P_i(|f|, n) = |f|(q_{i-1}(n))\]
and
\[P_d(f,x) = q_d(P_{d-1}(|f|,n)) \geq P(|f|, n).\]

The following pseudo-code provides a functional $B(f,x)$ that computes
the space bound of the OTM $M$ with input $f$ and $x$. The description
of $B(f,x)$ is written in imperative programming language style
pseudo-code. Note that $d$ is a fixed constant, it is easy to
transform this pseudo-code into functional algebra simply using $d$
levels of composition of the functional $L$.

\noindent
------------------------------------------------\\
{\bf input} $f$, $x$\\
$q_\mathrm{max} =\lambda$\\
{\bf for} $i=0$ {\bf to} $d-1$\\
\indent \indent $q_\mathrm{max} = L(f, 1^{q_i(|x|)})$\\
{\bf return } $1^{|\Ap(f, q_\mathrm{max})|}$\\
------------------------------------------------

As soon as we have the actual space bound $B(f,x)$ of the computation
of $M$ on input $f$ and $x$ and hence the bound of number of
(transition) steps $M^f(x)$ takes to run, we can define a functional
$Run_M$ similar to Kapron and Cook. Our functional $Run_M$ differs
from theirs mainly in two aspects. One is that ours keeps track of
only the encoding of the instantaneous description of the Turing
machine at the current computation step, while theirs keeps track of
the encoding of the entire history of the computation of the Turing
machine. The other is that our $Run_M$ uses bounded recursion, while
theirs uses bounded recursion on notation. The techniques used in
transforming Turing machine transition function to functional algebra
are standard, though tedious.  $Run_M(f, x, y)$ recurses on the value
of $y$ and $Run_M(f, x, B(f,x))$ is the instantaneous description at
the time $M(f, x)$ halts and
\[Run_M(f, x, B(f,x))\leq B(,x)\]
for all $f$ and $x$.  \qed\end{proof}

\newcommand{\nextid}{\mathrm{nextID}}
\newcommand{\nextq}{\mathrm{nextQ}} \newcommand{\id}{\mathrm{ID}}

%% file: section4.tex
\section{Resource Bounds}\label{se:4}
In the initial development of a theory of resource-bounded measure
\cite{Lutz:RBM}, a list of examples of resource-bounds were given
based on an oracle Turing machine model of type-2 computation that is
not known to be robust. In this section, we axiomatize the definition
of a resource bound by adapting the axioms of Mehlhorn's basic
feasible functionals and verify that most extensively used resource
bounds are indeed resource bounds under this definition.

\begin{definition}
A {\em resource bound} is a class $\Delta$ of functionals of type no
more than $2$ that is closed under $\BFF$.
\end{definition}

\begin{theorem}\label{th:4_1}
Let $i\in\Z^+$. $\p_i = \BFF_i$ is a resource bound.
\end{theorem}
\begin{proof}[Proof of Theorem \ref{th:4_1}]
Note that by definition $\BFF_i = \BFF(\{x\mapsto1^{g_i(|x|)}\})$.
Since $\BFF(\BFF(X))=\BFF(X)$ for all $X$, $\BFF(\BFF_i) =
\BFF_i$. Therefore $\BFF_i$ is a resource bound.  \qed\end{proof}

Let $K^k$ be the canonical $\Sigma^\P_k$-complete language
\cite{BaDiGa95}.  Let $\chi_k$ be the characteristic function of
$K^k$.
\begin{definition}
Let $i\in\Z^+$ and let $k\geq 2$.  $\Delta^{\p_i}_k =
\BFF_i(\{\chi_{k-1}\})$.
\end{definition}
\begin{theorem}\label{th:4_2}
Let $i\in\Z^+$ and let $k\geq 2$.  $\Delta^{\p_i}_k$ is a resource
bound.
\end{theorem}
\begin{proof}[Proof of Theorem \ref{th:4_2}]
Note that by definition $\Delta^{\p_i}_k =
\BFF(\{x\mapsto1^{g_i(|x|)}, \chi_{k-1}\})$.  Since
$\BFF(\BFF(X))=\BFF(X)$ for all $X$, $\BFF(\Delta^{\p_i}_k) =
\BFF(\BFF(\{x\mapsto1^{g_i(|x|)}, \chi_{k-1}\})) =
\BFF(\{x\mapsto1^{g_i(|x|)}, \chi_{k-1}\}) = \Delta^{\p_i}_k$.
Therefore $\Delta^{\p_i}_k$ is a resource bound.  \qed\end{proof}

\begin{definition}
Let $i\in\Z^+$. $\pispace = \BFSF_i$.
\end{definition}
\begin{theorem}\label{th:4_3}
Let $i\in\Z^+$. $\BFF(\pispace) = \pispace$, i.e., $\pispace$ is a
resource bound.
\end{theorem}
\begin{proof}[Proof of Theorem \ref{th:4_3}.]
We prove the equivalence for $i=2$. Since each of the $G_i$ is closed
under composition, the proof readily extends to all $i\in\Z^+$.

It suffices to show that $\BFF(\ptwospace) \subseteq \ptwospace$,
i.e., $\ptwospace$ is closed under functional composition, expansion,
and limited recursion on notation.

{\bf Functional composition}

Let $H, G_1, \dots, G_l \in \ptwospace$.  Let $F$ be defined from $H,
G_1, \dots, G_l$ by functional composition, i.e., for all $\vec f
\in \ptwospace$ and $\vec x \in \strings$,
\[F(\vec f, \vec x) = H(\vec f, 
                         G_1(\vec f, \vec x), 
                         \dots, 
                         G_l(\vec f, \vec x)).\]
Now, we show that $F\in \ptwospace$.

Since $H, G_1, \dots, G_l \in \ptwospace\cap\Ttwo$, there exist oracle
Turing machines $M_H$, $M_1, \dots, M_l$ and second-order polynomial
space bounds $P_H, P_1, \dots, P_l: ((\N\rightarrow\N)\times \N)
\rightarrow \N$ respectively.

Consider the following oracle Turing machine $M$.

---------------------------\\ 
\indent{\bf input} $\vec x$\\ 
\indent{\bf oracle} $\vec f$\\ 
\indent{\bf for} $i$ := $1$ to $l$
{\bf do}\\ 
\indent\indent $u_i$ := $M_i^{\vec f}(\vec x)$\\ 
\indent{\bf end for}\\ 
\indent{\bf output} $M_H^{\vec f}(u_1,\dots, u_l)$\\ 
\indent---------------------------\\

It is clear that $M$ computes $F$.

In the for loop of $M$, the $i$th iteration uses space at most
$P_i(|\vec f|, |\vec x|)$. The total space used in the for loop is
\[O\left( \sum_{i=1}^l P_i(|\vec f|, |\vec x|) \right).\]

The lengths of $u_1 = G_1(\vec f, \vec x)$, $\dots$, and $u_l =
G_l(\vec f, \vec x)$ are bounded by $P_1(|\vec f|, |\vec x|)$,
$\dots$, and $P_l(|\vec f|, |\vec x|)$ respectively.  So the space use
in the simulation of $M_H$ is bounded by
\[P_H(|\vec f|, P_1(|\vec f|, |\vec x|), \dots, P_l(|\vec f|, |\vec x|)).\]
The total space used in the computation of $F$ is
\[O\left( P_H(|\vec f|, P_1(|\vec f|, |\vec x|), \dots, P_l(|\vec f|, |\vec x|))
+\sum_{i=1}^l P_i(|\vec f|, |\vec x|) \right),\] which is a
second-order polynomial in $|\vec f|$ and $|\vec x|$ as both $|\vec
f|$ and $|\vec x|$ are arbitrary.

{\bf Expansion} Let $G\in \ptwospace$. Let $M_G$ be the oracle Turing
machine for $G$ and let $P_G$ be the corresponding second-order
polynomial that bounds the space for $M_G$.

Let $F$ be defined from $G$ by expansion, i.e., for all $\vec f, \vec
g, \vec x, \vec y$,
\[F(\vec f, \vec g, \vec x, \vec y) = G(\vec f, \vec x).\]

Consider the following oracle Turing machine $M$.

---------------------------\\ 
\indent{\bf input} $\vec x$, $\vec y$\\ 
\indent{\bf oracle} $\vec f$, $\vec g$\\ 
\indent{\bf output}
$M_G^{\vec f}(\vec x)$\\ 
\indent---------------------------\\

It is clear that for all $\vec f, \vec g, \vec x, \vec y$
\[M^{\vec f, \vec g}(\vec x, \vec y) =  G(\vec f, \vec x)= F(\vec f, \vec g, \vec x, \vec y).\]
It is easy to verify that the second-order polynomial $P$ defined by
\[P(|\vec f|, |\vec g|, |\vec x|, |\vec y|) =2 P_G(|\vec f|, |\vec x| + |\vec  y|) \]
bounds the space of $M$.

{\bf Limited recursion on notation}

Let $G, H, K\in \ptwospace$. Let $M_G, M_H, M_K$ be the oracle Turing
machines that computes $G, H, K$ respectively. Let $P_G, P_H, P_K$ be
their corresponding space bound respectively.

Let $F$ be defined from $G, H, K$ by limited recursion on notation.

Consider the following oracle Turing machine $M$.

---------------------------\\ 
\indent{\bf input} $\vec x$, $w$\\ 
\indent{\bf oracle} $\vec f$\\ 
\indent$u$ := $M_G^{\vec f}(\vec x)$\\ 
\indent{\bf for} $i$ := $1$ to $|w|-1$ {\bf do}\\ 
\indent\indent$u$ := $M_H^{\vec f}(\vec x, w[0..i], u)$\\ 
\indent{\bf end for}\\ 
\indent{\bf output} $u$\\ 
\indent---------------------------\\

Clearly, Turing machine $M$ computes $F$ (using the iterative
expansion of the recursion).  Now, we show that $M$ runs in space that
is bounded by a second-order polynomial.

The third line of code in $M$ uses space bounded by $P_G(|\vec f|,
|\vec x|)$.

At $i$th iteration of the for loop, the fifth line in $M$ computes the
value of $F(\vec f, \vec x, w[0..i])$ and uses space bounded by
$P_H(|\vec f|, |F(\vec f, \vec x, w[0..i-1])|, |w[0..i]|)$.  By the
restriction in the definition of limited recursion on notation and the
monotonicity of $P_K$, at any time during the computation
\[|F(\vec f, \vec x, w[0..i-1])|\leq P_K(|\vec f|, |\vec x|, |w|).\]
Thus the space used at $i$th iteration of the for loop is bounded by
\[P_H(|\vec f|,  P_K(|\vec f|, |\vec x|, |w|), |w[0..i]|).\]
The auxiliary space used for the for loop is bounded by $c\cdot
\log|w|\leq c\cdot |w|$, where $c>0$ is some universal constant.  The
total amount of space used by $M$ is bounded by
\[|w|\cdot P_H(|\vec f|,  P_K(|\vec f|, |\vec x|, |w|), |w[0..i]|) + c\cdot |w|,\]
which is a second-order polynomial in $|\vec f|$, $|\vec x|$, and
$|w|$.  
\qed\end{proof}

%% file: section5.tex
\section{Adequacy for Measure}\label{se:5}
The general theory of resource-bounded measurability and measure
developed in \cite{Lutz:RBM} consists of the basic definitions,
reviewed below, and proofs that the resulting $\Delta$-measure and
measure in $R(\Delta)$ have the fundamental properties of a measure
(e.g., additivity, measurability of measure-$0$ sets, etc.).  The main
shortcoming of the list-of-examples approach is evident in these
proofs: Each time that a functional is asserted to be
$\Delta$-computable, it is {\em incumbent on the reader to check} that
this holds for each of the infinitely many resource bounds $\Delta$ in
the list.

Our main task in the present section is to re-prove these theorems in
a more satisfactory manner.  Our proofs here assume only that $\Delta$
is a resource bound, as defined in section 3, and they {\em
  explicitly} prove that the relevant functionals are
$\Delta$-computable, using only the axioms (closure properties)
defining resource bounds.

To put the matter simply, the proofs in \cite{Lutz:RBM} are
measure-theoretically rigorous, but their generality is tedious (for
the conscientious reader) and limited.  Our contribution here is to
make these proofs and the scope of their validity explicit.  For this
reason, the proofs given in the present section focus on the
$\Delta$-computability of various type-two functionals, referring to
\cite{Lutz:RBM} for the non-problematic, measure-theoretic parts of
the proofs.  We first review the definitions necessary for the
development of a resource-bounded measure.

A {\em probability measure} on $\C$ is a function
$\nu:\strings\rightarrow[0,1]$ such that $\nu(\lambda)=1$ and, for all
$w\in\strings$, $\nu(w)=\nu(w0)+\nu(w1).$ For strings
$v,w\in\strings$, if $\nu(w)>0$, we write $\nu(v|w)$ for the
conditional probability of $v$ given $w$. The {\em uniform probability
  measure} is $\mu$ such that $\mu(w) = 2^{-|w|}$ for all
$w\in\strings$.

Let $\nu$ be a probability measure on $\C$. A $\nu$-{\em martingale}
is a function $d:\strings\rightarrow[0,\infty)$ with the property that
  for all $w\in\strings$,
\[d(w)\nu(w) = d(w0)\nu(w0)+d(w1)\nu(w1).\]
We use $\mathbf{1}$ for the {\em unit martingale} defined by
$\mathbf{1}(w)=1$ for all $w\in\strings$, which is a $\nu$-martingale
for every probability measure $\nu$.

\begin{definition}
Let $\nu$ be a probability measure. Let $d$ be a $\nu$-martingale. Let
$A\subseteq \strings$.  We say that $d$ {\em covers} $A$ if there is
an $n\in\N$ such that $d(A[0..n-1])\geq 1$.  We say that $d$ {\em
  succeeds} on $A$ if $\limsupn d(A[0..n-1])=\infty$ We say that $d$
{\em succeeds strongly} on $A$ if $\liminfn d(A[0..n-1])=\infty$ The
{\em set covered by} $d$ (the {\em unitary success set}) is $S^1[d]
=\myset{A}{d\text{ covers }A}$.  The {\em success set} of $d$ is
$S^\infty[d] =\myset{A}{d \text{ succeeds on } A}$.  The {\em strong
  success set} of $d$ is $S^\infty_\str[d] =\myset{A}{d \text{
    succeeds strongly on } A}$.
\end{definition}

We use real-valued functions (probability measures, martingales, etc.)
on discrete domains of natural numbers $\N$ and strings $\strings$
extensively. Let $D$ be a discrete domain.  A {\em computation} of a
function $f:D\rightarrow\R$ is a function $\hat f:\N\times
D\rightarrow\Q$ such that, for all $r\in\N$ and $x\in D$, $|\hat
f(r,x) - f(x)|\leq 2^{-r}$. In this expression, $r$ may be thought of
as the \emph{precision} parameter of the computation. For such a
function $f$, there is a unique computation $\hat f$ of $f$ such that
$\hat f_r(x) = a\cdot 2^{-r}$ for some integer $a$ for all $r\in \N$
and $x\in D$.  We call this particular $\hat f$ the {\em canonical
  computation} of $f$.  Whenever a function $f$ is involved as a
parameter in the of a type-2 functional, the type-2 computation of the
functional operates on the canonical computation, $\hat f$.

\begin{definition}[Lutz \cite{Lutz:RBM}]
Let $\Delta$ be a resource bound. A $\Delta$-{\em probability measure}
on $\C$ is a probability measure $\nu$ on $\C$ such that $\nu$ is
$\Delta$-computable and there is a $\Delta$-computable function
$l:\N\rightarrow\N$ such that, for all $w\in\strings$, $\nu(w)=0$ or
$\nu(w)\geq 2^{-l(|w|)}$. We say that $\nu$ is {\em weakly positive},
if $\nu$ has the latter property.
\end{definition}

\begin{definition}[Lutz \cite{Lutz:AEHNC,Lutz:RBM}]
A {\em constructor} is a function $\delta:\strings\rightarrow\strings$
such that $x\prprefix \delta(x)$ for all $x\in\strings$. The {\em
  result} of $\delta$ is the unique language $R(\delta)$ such that
$\delta^k(\lambda)\prefix R(\delta)$.  If $\Delta$ is a resource
bound, then the {\em result class} $R(\Delta)$ of $\Delta$ is the set
of all languages $R(\delta)$ such that $\delta\in \Delta$.
\end{definition}

The martingale splitting operators defined by Lutz \cite{Lutz:RBM} are
instrumental in developing the general theory of resource-bounded
measurability and measure in complexity classes.

\begin{definition}[Lutz \cite{Lutz:RBM}]
Let $X^+$ and $X^-$ be disjoint subsets of $\C$, then a $\nu$-{\em
  splitting operator} for $(X^+, X^-)$ is a functional
$\Phi:\N\times\D_\nu\rightarrow \D_\nu\times \D_\nu,$ such that
$\Phi(r,d) =(\Phi_r^+(d), \Phi_r^-(d))$ has the following properties
for all $r\in\N$ and $d\in\D_\nu$.
\begin{enumerate}[(i)]
\item
$X^+\cap S^1[d]\subseteq S^1[\Phi_r^+(d)]$,
\item
$X^-\cap S^1[d]\subseteq S^1[\Phi_r^-(d)]$,
\item
$\Phi_r^+(d)(\lambda)+\Phi_r^-(d)(\lambda)\leq d(\lambda)+2^{-r}$.
\end{enumerate}
If $\Delta$ is a resource bound, a $\Delta$-$\nu$-splitting operator
for $(X^+, X^-)$ is a $\nu$-splitting operator for $(X^+, X^-)$ that
is $\Delta$-computable.  Let $X\subseteq \C$.  A
$\Delta$-$\nu$-measurement of $X$ is a $\Delta$-$\nu$-splitting
operator for $(X^+, X^-)$.  A $\nu$-measurement of $X$ in $R(\Delta)$
is a $\Delta$-$\nu$-splitting operator for $(R(\Delta)\cap X^+
,R(\Delta)-X)$.  If $\Phi$ is a $\nu$-splitting operator, then we
write $\Phi^+_\infty= \inf_{r\in\N}\Phi_r^+(\mathbf{1})(\lambda),$
$\Phi^-_\infty= \inf_{r\in\N}\Phi_r^-(\mathbf{1})(\lambda).$
\end{definition}

We now can define the resource-bounded measurabilities.

\begin{definition}
A set $X\subseteq\C$ is $\nu$-measurable in $R(\Delta)$, and we write
$X\in\F_{R(\Delta)}^\nu$, if there exists a $\nu$-measurement $\Phi$
of $X$ in $R(\Delta)$. In this case, the $\nu$-measure of $X$ in
$R(\Delta)$ is the real number $\nu(X|R(\Delta))=\Phi^+_\infty$.
($\nu(X|R(\Delta))$ does not depend on the choice of $\Phi$
\cite{Lutz:RBM}.)
\end{definition}

\begin{definition}
A set $X\subseteq\C$ is $\Delta$-$\nu$-measurable, and we write
$X\in\F_\Delta^\nu$, if there exists a $\Delta$-$\nu$-measurement
$\phi$ of $X$. In this case, the $\Delta$-$\nu$-measure of $X$ is the
real number $\nu_\Delta(X)=\Phi_\infty^+$.  ($\nu_\Delta(X)$ does not
depend on the choice of $\Phi$ \cite{Lutz:RBM}.)
\end{definition}

In the rest of this paper, we refer to \cite{Lutz:RBM} liberally
whenever a claim was already proved. 

\begin{theorem}[Measure Conservation Theorem]\label{th:5_1}
Let $\Delta$ be a resource bound.
If $w\in\strings$ and $d$ is a $\Delta$-$\nu$-martingale
such that $\C_w\cap R(\Delta)\subseteq S^1[d]$, then $d(\lambda)\geq \nu(w)$.
\end{theorem}
\begin{proof}
We define a functional 
$$E:\D_\nu\times \N\times \strings\rightarrow (\strings\rightarrow\strings)$$
that maps $\nu$-martingales to constructors.

Let $a:\strings\times \N\rightarrow \N$ be such that $a(x, m) = |x|+m+2$. It is clear that $a$ is $\BFF$.

Let $E$ be such that for all $d \in \D_\nu$, $m\in \N$, $w\in \strings$, and $x\in\strings$,
\[
E(d, m, w)(x) =
\begin{cases}
w&\text{if }x\prprefix w\\
x0&\text{if }\hat{d}_{a(x,m)}(x0)\leq \hat{d}_{a(x,m)}(x1)\text{ and not } x\prprefix w\\
x1&\text{otherwise,}
\end{cases}\]
Note that $\hat{d}$ is the canonical computation of $d$ and it is clear that $E$ is $\BFF$ and
$E(d, m, w)$ is a constructor.

Let $w\in\strings$ and let $d$ be a $\Delta$-$\nu$-martingale such that
$d(\lambda)<\nu(w)$. Then  for every prefix $w'\prefix w$,
there exists a constant $m\in\N$ such that $d(w)\leq 1- 2^{1-m}$ \cite{Lutz:RBM}.
Let $\delta = E(d, m, w)$. By the proof of Lemma 3.4 in \cite{Lutz:RBM},
$R(\delta)\notin S^1[d]$. Since $d\in\Delta$ and $\Delta$ is a resource bound hence closed under $\BFF$,
$\delta\in\Delta$.
\qed\end{proof}
\begin{lemma}\label{lm:5_2}
Let $\mathrm{ADD}:\D_\nu\times \D_\nu\rightarrow \D_\nu$ be such that
for all $d', d''\in\D_\nu$, $\mathrm{ADD}(d',d'') = d'+d''$.
Then $\mathrm{ADD}\in\BFF$.
\end{lemma}
\begin{proof}
We write $d$ for $\mathrm{ADD}(d',d'')$. Then
\[\hat{d}_r(w) = \hat{d'}_{r+1}(w) + \hat{d''}_{r+1}(w).\]
Note that addition of two rational numbers in binary expansion is $\BFF$.
Since $|\hat{d'}_{r+1}(w) - d'(w) |\leq 2^{-r-1}$ and $|\hat{d''}_{r+1}(w) - d''(w) |\leq 2^{-r-1}$,
$|\hat{d}_r(w) - (d'(w)+d''(w))|\leq 2^{-r}$ and $\hat{d}$ is the canonical computation of $d$.
Therefore $\mathrm{ADD}$ is $\BFF$.
\qed\end{proof}
\begin{lemma}[Lutz \cite{Lutz:RBM}]\label{lm:5_3}
Let $\Delta$ be a resource bound.
Let $X\subseteq \C$. If $\Phi$ and $\Psi$ are $\nu$-measurements of $X$ in $R(\Delta)$, then for all
$j, k\in\N$,
\[\Phi_j^+(\mathbf{1})(\lambda) + \Psi_k^+(\mathbf{1})(\lambda) \geq 1.\]
\end{lemma}
\begin{proof}
Assume the hypothesis, let $j,k\in\N$, and let
\[d=\mathrm{ADD}(\Phi_j^+(\mathbf{1}),\Psi_k^+(\mathbf{1})).\]
Since $\mathrm{1}$ is $\BFF$, both $\Phi$ and $\Psi$ are $\Delta$,
and $\Delta$ is a resource bound and closed under $\BFF$,
$d\in\Delta$. The rest of the proof is identical to the proof of
Lemma 4.1 in \cite{Lutz:RBM}.
\qed\end{proof}
\begin{lemma}[Lutz \cite{Lutz:RBM}]\label{lm:lutz4_5}
Let $X\subseteq\C$ and let $\Delta$ be a resource bound.
\begin{enumerate}[1.]
\item
If $X$ is $\Delta$-$\nu$-measurable, then $X$ is $\nu$-measurable in $R(\Delta)$ and $\nu(X|R(\Delta))=\nu_\Delta(X)$.
\item
$X$ is $\nu$-measurable in $R(\Delta)$ if and only if $X\cap R(\Delta)$ is $\nu$-measurable in $R(\Delta)$,
in which case $\nu(X\mid R(\Delta)) =\nu(X\cap R(\Delta)\mid R(\Delta)$.
\end{enumerate}
\end{lemma}
\begin{theorem}[Lutz \cite{Lutz:RBM}]\label{th:5_5}
Let $\Delta$ be a resource bound. Let $X\subseteq \C$.
\begin{enumerate}
\item
If $X$ is $\nu$-measurable in $R(\Delta)$, then $\nu(X| R(\Delta))$ is $\Delta$-computable.
\item
If $X$ is $\Delta$-$\nu$-measurable, then $\nu_\Delta(X)$ is $\Delta$-computable.
\end{enumerate}
\end{theorem}
\begin{proof}
We prove 1, since 2 follows by Lemma \ref{lm:lutz4_5}.

Let $S_\nu$ be the set of all $\nu$-splitting operators.  Note that
$\mathbf{1}$ is $\BFF\subseteq\Delta$.  Let $\Phi$ be a
$\nu$-measurement for $X$ in $R(\Delta)$.  Then $\Phi\in\Delta$. Since
$\Delta$ is a resource bound and closed under $\BFF$, $r\mapsto
\hat{\Phi}_{r,r}^+(\mathbf{1})$ is in $\Delta$.  Let
\[f(r) = \hat{\Phi}_{r+1,r+1}^+(\mathbf{1})(\lambda).\]
Then $f\in\Delta$ and by the proof of Theorem 4.7 in \cite{Lutz:RBM},
$f$ is the canonical computation of the real value of $\nu(X|R(\Delta))$.
\qed\end{proof}

We now proceed towards the proof that cylinders are measurable. First,
we prove lemmas that are useful for the proof of Theorem
\ref{th:5_11}.

\begin{lemma}[Regularity Lemma]\label{lm:lutz3.5}
Let $\Delta$ be a resource bound. There is a functional
\[
\func{\Lambda}{{\cal D}_\nu}{{\cal D}_{\nu}}
\]
with the following properties.
\begin{enumerate}[1.]
\item For all $d \in {\cal D}_\nu$, $\Lambda(d)$ is a regular
  $\nu$-martingale such that $\Lambda(d)(\lambda) = d(\lambda)$ and
  $S^1[d] \subseteq S^1[\Lambda(d)]$.
\item $\Lambda(\mathbf{1}) = \mathbf{1}$.
\item If $\nu$ is a $\Delta$-probability measure on $\C$, then
  $\Lambda$ is $\Delta$-computable.
\end{enumerate}
\end{lemma}

The proof of the regularity lemma proceeds by defining a type-2
functional. We establish the computability properties of this
functional, first.

\include{pasting_lemma}

The following essential properties of the Robin Hood function
$rh_\alpha$ are routine to verify.

\begin{enumerate}
\item[1.] The transformation $rh_\alpha$ is a continuous, piecewise
linear mapping from $D_\alpha$ into $[0, \infty)^2$.
\item[2.] The transformation $rh_\alpha$ preserves $\alpha$-weighted
averages, i.e., $m_\alpha(rh_\alpha(s, t)) = m_\alpha(s, t)$ for all
$(s, t) \in D_\alpha$.
\item[3.] The transformation $rh_\alpha$ maps $H_\alpha$ into $[1,
\infty)^2$.  That is, if the average $m_\alpha(s, t)$ is at least $1$,
then $rh_\alpha$ ``steals from the richer and gives to the poorer'' of
$s$ and $t$ so that both $rh^{(0)}_\alpha(s, t)$ and
$rh^{(1)}_\alpha(s, t)$ are at least $1$.
\item[4.] For all $(s, t) \in D_\alpha$, $rh^{(0)}_\alpha(s, t) \geq
\min \{1, s \}$ and $rh^{(1)}_\alpha(s, t) \geq \min \{1, t \}$.  That
is, the transformation $rh_\alpha$ never ``steals'' more than the
excess above $1$.
\item[5.] The transformation $rh_\alpha$ leaves points of $[0, 1]^2$ unchanged.
\end{enumerate}


A $\nu$-martingale is {\it regular} if, for all $v, w \in \strings$, if
$\nu(v) \geq 1$ and $v \sqsubseteq w$, then $\nu(w) \geq 1$.  It is
often technically convenient to have a uniform means of ensuring that
martingales are regular.  The following lemma provides such a
mechanism.  Let $\Delta$ be a resource bound, as specified in
section \ref{se:2}, and let $\nu$ be a probability measure on $\C$.

\begin{proof}[Proof of Lemma \ref{lm:lutz3.5}]
Using the Robin Hood function, we define the functional
$\func{\Lambda}{{\cal D}_\nu}{{\cal D}_\nu}$ as follows.  For $d \in
{\cal D}_\nu$, we define the $\nu$-martingale $\Lambda(d)$ by the
following recursion.  (In all clauses, $w \in \strings$ and $b \in
\{0, 1\}.$)
\begin{enumerate}
\item[(i)] $\Lambda(d)(\lambda) = d(\lambda)$.
\item[(ii)] If $\nu(w) = 0$ or $\nu(wb \mid w) \in \{0, 1\}$, then
$\Lambda(d)(wb) = \Lambda(d)(w)$.
\item[(iii)] If $\nu(w) > 0$ and $0 < \nu(wb \mid w) < 1$, then
\[
\Lambda(d)(wb) = rh^{(b)}_{\nu(w0 \mid w)} (g_0(w), g_1(w)),
\]
where $g_b(w) = \Lambda(d)(w) - d(w) + d(wb)$.
\end{enumerate}

Let $\hat{d}$ be a computation of $d$ and $\hat{\nu}: \strings \times
0^{\N} \to \Q$ be the function testifying that $\nu$ is $\Delta$
computable. Since $\nu$ is $\Delta$-computable, we have a function $l:
\N \to \N$ in $\Delta$ such that for every string $w$, $\nu(w) >
2^{-l(|w|)}$. The functional $\hat{\Lambda}: \mathcal{D}_\nu \times
\strings \times 0^\N \to \Q$ is a BFF computation of $\Lambda$.
\begin{enumerate}
\item $\hat{\Lambda}(\hat{d})(\lambda,0^n) = \hat{d}(\lambda, 0^n)$.
\item If $\hat{\nu}(w) < 2^{-l(|w|)}$ or $\hat{\nu}(w0,0^n) <
  2^{-l|w0|}$ or $\hat{\nu}(w1,0^n) < 2^{-l|w1|}$, then
  $\hat{\Lambda}(d)(wb, 0^{n}) = d(wb, 0^{n})$.
\item Otherwise, $\hat{\Lambda}(\hat{d})(wb,0^n) =
  \hat{rh}^{(b)}_{\hat{\nu}(w0|w)} (\hat{g}_0(w,0^n), g_1(w,0^n))$.
\end{enumerate}

Note that if $\nu$ is a strongly positive probability measure, then
$\nu(b|w) = 1$ if and only if $\nu(w\overline{b} | w) = 0$. Assuming
$\nu(w) > 0$, we have that $\nu(w\overline{b} | w) = 0$ if and only if
$\nu(w\overline{b}) = 0$, i. e. $\nu(\overline{b}) < 2^{-l(|w|)}$. Thus
step 2 correctly approximates step 2 of $\Lambda$.

It is now routine (if tedious) to verify that $\Lambda$ has the
desired properties.
\qed\end{proof}

\begin{lemma}\label{lm:bffconditional}
Let type-2 functional
$B:(\strings\rightarrow\R)\times(\N\rightarrow\N)\rightarrow
(\strings\times\strings\rightarrow\R)$ be such that for every weakly
positive probability measure $\nu:\strings\rightarrow[0,1]$,
$l:\N\rightarrow\N$, and $w, v\in\strings$,
\[
B(\nu, l)(w,v)=
\begin{cases}
\nu(w|v)&\text{if }v\prefix w\text{ and }\nu(w)\geq
2^{-l(|w|)}\\ 1&\text{if }w\prefix v\text{ and }\nu(w)\geq
2^{-l(|w|)}\\ 0&\text{otherwise.}
\end{cases}
\]
Then $B$ is $\BFF$ over all weakly positive probability measure $\nu$
and all $l:\N\rightarrow\N$.
\end{lemma}

\begin{theorem}[Lutz \cite{Lutz:RBM}]\label{th:5_11}
Let $\Delta$ be a resource bound.  If $\nu$ is a $\Delta$-probability
measure on $\C$, then for each $w\in\strings$, the cylinder $\C_w$ is
$\Delta$-$\nu$-measurable, with $\nu_\Delta(\C_w)=\nu(w)$.
\end{theorem}
\begin{proof}
Assume the hypothesis, and let $w\in\strings$. We prove this lemma in
two cases.

Case 1: $\nu(w)=0$.  Let $\Phi:\N\times\D_\nu\rightarrow\D_\nu\times
\D_\nu$ be such that for each $r\in\N$, $d\in\D_\nu$, and
$v\in\strings$,
\begin{align*}
\Phi_r^+(d)(v)&=\begin{cases} 1&\text{if }w\prefix
v\\ 0&\text{otherwise,}
\end{cases}\\
\Phi_r^-(d)(v)&=d(v).
\end{align*}
It is clear that $\Phi$ is $\BFF$ and hence $\Delta$-computable and
that $\Phi$ is a $\nu$-splitting operator. By the proof of Lemma 4.8
in \cite{Lutz:RBM}, for all $r\in\N$ and $d\in\D_\nu$, $\Phi$ has the
following properties: (i)
$S^1[d]\cap\C_w\subseteq\C_w=S^1[\Phi_r^+(d)]$; (ii)
$S^1[d]-\C_w\subseteq S^1[d]=S^1[\Phi_r^-(d)]$; (iii)
$\Phi_r^+(d)(\lambda)+\Phi_r^-(d)(\lambda)=d(\lambda)$.  Therefore,
$\Phi$ is a $\Delta$-$\nu$-measurement of $\C_w$.  It can be shown
that $\nu_\Delta(w) = 0$ \cite{Lutz:RBM}. Note that $w\neq \lambda$,
since $\nu(\lambda)=1$ and $\nu(w)=0$.

Case 2: $\nu(w)>0$.  Let $\Psi':\D_\nu\rightarrow\D_\nu\times\D_\nu$
be such that for each $d\in\D_\nu$ and $v\in\strings$,
\begin{align*}
\Psi'^+(d)(v)&=\begin{cases} d(w)\nu(w|v)&\text{if }v\prefix
w\\ d(v)&\text{if }w\prefix v\\ 0&\text{otherwise,}
\end{cases}\\
\Psi'^-(d)(v)&=d(v)-\Psi'^+(d)(v).
\end{align*}
Since $\nu$ is a $\Delta$-probability measure, there exists
$l:\N\rightarrow\N$ such that $l$ is $\Delta$-computable and for all
$w\in\strings$, $\nu(w)=0$ or $\nu(w)\geq 2^{-l(|w|)}$.  Then
\begin{align*}
\Psi'^+(d)(v)&=\begin{cases} d(w)B(\nu, l)(w,v)&\text{if }v\prefix
w\\ d(v)&\text{if }w\prefix v\\ 0&\text{otherwise,}
\end{cases}\\
\end{align*}
where $B$ is the functional defined in Lemma \ref{lm:bffconditional}.
Note that $\Psi'$ is $\BFF$ over type-1 input $d$ and $\nu(\cdot
|\cdot)$.  By Lemma \ref{lm:bffconditional}, we have that $\Psi'$ is
$\BFF$ over probability measure $\nu$ and $\nu$-martingale $d$.

Let $\Psi:\N\times\D_\nu\rightarrow\D_\nu\times\D_\nu$ be such that
$\Psi(r,d)=\Psi'(\Lambda(d)),$ where $\Lambda$ is the functional from
the Regularity Lemma.  Since $\nu$ is a $\Delta$-probability measure,
$\Lambda$ is $\Delta$-computable.  Since $\Psi'$ is $\BFF$ and $d,
\nu, l, \Lambda\in\Delta$, $\Psi$ is $\Delta$-computable since
$\Delta$ is closed under $\BFF$.  The rest of the proof is to
establish that $\Psi$ is a $\Delta$-$\nu$-measurement of $\C_w$ and
$\nu_\Delta(\C_w)=\nu(w)$, which follows directly from the proof of
Lemma 4.8 in \cite{Lutz:RBM}.
\end{proof}

\begin{definition}
Let $R\subseteq \C$. An {\em algebra} on $R$ is a collection $\F$ of subsets of $\C$ with the
following properties.
\begin{enumerate}[(i)]
\item
$R\in\F$.
\item
If $X\in\F$, then $X^c\in\F$.
\item
If $X,Y\in\F$, then $X\cup Y\in\F$.
\end{enumerate}
If $\F$ is an algebra on $R$, then a {\em subalgebra} of $\F$ on $R$ is a set $\mathcal{E}$
that is also an algebra on $R$.
\end{definition}
\begin{theorem}[Lutz \cite{Lutz:RBM}]\label{th:Rdeltaalgebra}
Let $\Delta$ be a resource bound.
$\F_{R(\Delta)}^\nu$ is an algebra on $R(\Delta)$. For $X, Y\in \F_{R(\Delta)}^\nu$, we have
\[\nu(X^c|R(\Delta)) = 1-\nu(X|R(\Delta))\]
and
\[\nu(X\cup Y|R(\Delta)) =\nu(X|R(\Delta)) +\nu(Y|R(\Delta)) - \nu(X\cap Y|R(\Delta)).\]
\end{theorem}
Note that this lemma does not have any computability requirement on $\nu$.
It is possible that $\F_{R(\Delta)}^\nu = \{\varnothing, \C\}$.
\begin{proof}
Let $\Phi$ be a $\nu$-measurement of $X$ in $R(\Delta)$.
Let
\[\Psi:\N\times\D_\nu\rightarrow\D_\nu\times \D_\nu\]
be such that
\[\Psi(r, d) = (\Phi_r^-(d), \Phi_r^+(d)).\]
Since $\Phi$ is $\Delta$-computable and $\Psi$ can be defined from $\Phi$ with projection and composition,
$\Psi$ is $\Delta$-computable. Then by Theorem 4.12 in \cite{Lutz:RBM},
$\Psi$ is a $\nu$-measurement of $X^c$ in $R(\Delta)$, i.e., $X^c\in \F_{R(\Delta)}^\nu$ and
$\nu(X^c|R(\Delta)) = 1-\nu(X|R(\Delta))$.

Now, let $X,Y\in \F_{R(\Delta)}^\nu$ and let $\Phi$ and $\Psi$ be $\nu$-measurements of $X$ and $Y$, respectively,
in $R(\Delta)$. For each $a, b\in\{+, -\}$, let
\[\Theta[ab]:\N\times\D_\nu\rightarrow\D_\nu\]
by
\[\Theta[ab](r,d) = \Psi_{r+2}^b(\Phi^a_{r+1}(d)).\]
Note that $\Theta[ab]$ is defined from $\Psi$ and $\Phi$ by $\BFF$.
Therefore, for each $a,b\in\{+, -\}$, $\Theta[ab]$ is $\Delta$-computable,
since $\Delta$ is closed under $\BFF$.
Let
\begin{align*}
\Theta[X\cap Y]&=(\Theta[++]+\Theta[-+], \Theta[+-]+\Theta[--]),\\
\Theta[X\cup Y]&=(\Theta[++]+\Theta[-+], \Theta[+-]+\Theta[--]).\\
\end{align*}
Then $\Theta[X\cap Y]$ and $\Theta[X\cup Y]$ are $\BFF$ in $\Theta[ab]$ ($a,b\in\{+,-\}$),
$\Theta[X\cap Y], \Theta[X\cup Y]$ are $\Delta$-computable.
Then by the proof of Theorem 4.12 in \cite{Lutz:RBM},
$\Theta[X\cap Y]$ and $\Theta[X\cup Y]$ are $\nu$-measurements of $X\cap Y$ and $X\cup Y$  in $R(\Delta)$ respectively,
and thus $X\cap Y$ and $X\cup Y$ are $\nu$-measurable in $R(\Delta)$.
The second identity in the theorem also follows from Theorem 4.12 in \cite{Lutz:RBM}.
\qed\end{proof}
\begin{corollary}[Lutz \cite{Lutz:RBM}]\label{co:algebra}
Let $\Delta$ be a resource bound.
Let $X, Y\in\F_{R(\Delta)}^\nu$.
\begin{enumerate}[1.]
\item
If $X\cap Y\cap R(\Delta)=\varnothing$, then
\[\nu(X\cup Y \mid R(\Delta)) =\nu(X\mid R(\Delta))+\nu(Y\mid R(\Delta)).\]
\item
If $X\cap R(\Delta)\subseteq Y$, then $\nu(X\mid R(\Delta))\leq \nu(Y\mid R(\Delta))$.
\end{enumerate}
\end{corollary}
\begin{theorem}\label{th:5_14}
Let $\Delta$ be a resource bound. $\F_\Delta^\nu$ is an algebra over $\C$.
For $X, Y\in\F_\Delta^\nu$, we have
\[\nu_\Delta(X^c)= 1-\nu_\Delta(X)\]
and
\[\nu_\Delta(X\cup Y) = \nu_\Delta(X) + \nu_\Delta(Y)-\nu_\Delta(X\cup Y).\]
\end{theorem}
\begin{proof}
The proof is similar to that of Theorem \ref{th:Rdeltaalgebra}.
\qed\end{proof}
\begin{definition}
Let $\F$ be a subalgebra of $\F_{R(\Delta)}^\nu$ on $R(\Delta)$.
Then $\F$ is $\nu$-{\em complete} on $R(\Delta)$ if,
for all $X,Y\subseteq \C$,
if $X\subseteq Y\in\F$ and $\nu(Y|R(\Delta))=0$, then $X\in\F$.
\end{definition}
\begin{definition}
Let $\F$ be a subalgebra of $\F_\Delta^\nu$ on $\C$.
Then $\F$ is $\Delta$-$\nu$-{\em complete} if,
for all $X,Y\subseteq \C$,
if $X\subseteq Y\in\F$ and $\nu_\Delta(Y)=0$, then $X\in\F$.
\end{definition}
\begin{theorem}\label{th:5_15}
The algebra $\F_{R(\Delta)}^\nu$ is $\nu$-complete on $R(\Delta)$ and
the algebra $\F_\Delta^\nu$ is $\Delta$-$\nu$-complete.
\end{theorem}
\begin{proof}
We prove the case with $\F_{R(\Delta)}^\nu$. The other case is similar.
Let $X\subseteq Y\in\F_{R(\Delta)}^\nu$ and $\nu(Y|R(\Delta))=0$.
Let $\Phi$ be a $\nu$-measurement of $Y$ in $R(\Delta)$. Let
\[
\Psi:\N\times\D_\nu\rightarrow\D_\nu\times\D_\nu\\
\]
be such that
\[\Psi(r,d) = (\Phi_r^+(\mathbf{1}), d).\]
Note that $\Psi$ is $\BFF$ in $\Phi$ and $\mathbf{1}$.
Since $\Phi$ is $\Delta$-computable and $\Delta$ is a resource bound,
$\Psi$ is also $\Delta$-computable.
By the proof of Theorem 4.16 in \cite{Lutz:RBM},
$\Psi$ is a $\nu$-measurement of $X$ in $R(\Delta)$
and thus $X\in \F_{R(\Delta)}^\nu$.
\qed\end{proof}

We cannot hope to have countable additivity in this theory
\cite{Lutz:RBM}, however we have the following additivity property
over uniformly computable sequences.

\begin{definition}[Lutz \cite{Lutz:RBM}]
Let ${\cal F}$ be a subalgebra of ${\cal F}^\nu_{R(\Delta)}$ on
$R(\Delta)$.  A $\Delta$-{\it{sequence in}} ${\cal F}$ is a sequence
$(X_k \mid k \in \N)$ of sets $X_k \in {\cal F}$ for which there
exists a $\Delta$-computable functional
\[
\func{\Phi}{\N \times \N \times {\cal D}_\nu}{{\cal D}_\nu \times
{\cal D}_\nu}
\]
such that, for each $k \in \N$, $\Phi_k$ is a $\nu$-measurement of
$X_k$ in $R(\Delta)$.
\end{definition}

\begin{definition}[Lutz \cite{Lutz:RBM}]
Let ${\cal F}$ be a subalgebra of ${\cal F}^\nu_\Delta$ on
$\C$.  A $\Delta$-{\it{sequence in}} ${\cal F}$ is a sequence
$(X_k \mid k \in \N)$ of sets $X_k \in {\cal F}$ for which there
exists a $\Delta$-computable functional
\[
\func{\Phi}{\N \times \N \times {\cal D}_\nu}{{\cal D}_\nu \times
{\cal D}_\nu}
\]
such that, for each $k \in \N$, $\Phi_k$ is a
$\Delta$-$\nu$-measurement of $X_k$.
\end{definition}

\begin{definition}[Lutz \cite{Lutz:RBM}]
A functional
\[
\func{\Phi}{\N \times \N \times {\cal D}_\nu}{{\cal D}_\nu \times
{\cal D}_\nu}
\]
is $\Delta$-{\it{modulated}} if the sequences $(\Phi^a_{k, r}(d)(w)
\mid k \in \N)$, for $a \in \{+, -\}$, $r \in \N$, $d \in {\cal
D}_\nu$, and $w \in \strings$, are uniformly $\Delta$-convergent.
Equivalently, $\Phi$ is $\Delta$-modulated if there is a
$\Delta$-computable functional
\[
\func{\Gamma}{\N \times \N \times {\cal D}_\nu \times \strings}{\N}
\]
such that, for all $a \in \{+, -\}$, $t, r \in \N$, $d \in {\cal
D}_\nu$, $w \in \strings$, and $k \geq \Gamma(t, r, d, w)$,
\[
\left| \Phi^a_{k, r}(d)(w) - \Phi^a_{\infty, r}(d)(w) \right| \leq
2^{-t}
\]
where the limit $\Phi^a_{\infty, r}(d)(w) = \lim_{k \rightarrow
\infty} \Phi^a_{k, r}(d)(w)$ is implicitly assumed to exist.
\end{definition}

\begin{definition}[Lutz \cite{Lutz:RBM}]
Let ${\cal F}$ be a subalgebra of ${\cal F}^\nu_{R(\Delta)}$ on
$R(\Delta)$.
\begin{enumerate}[1.]
\item A {\it union} $\Delta$-{\it{sequence in}} ${\cal F}$ is a
sequence $(X_k \mid k \in \N)$ of sets $X_k \in {\cal F}$ for which
there exists a $\Delta$-modulated functional $\Phi$ such that
$\Phi^+_{k,r}(d)(w)$ is nondecreasing in $k$, $\Phi^-_{k, r}(d)(w)$ is
nonincreasing in $k$, and $\Phi$ testifies that
$\left(\bigcup^{k-1}_{j=0} X_j \mid k \in \N\right)$
is a $\Delta$-sequence in ${\cal F}$.
\item An {\it intersection} $\Delta$-{\it{sequence in}} ${\cal F}$
is a sequence $(X_k \mid k \in \N)$ of sets $X_k \in {\cal F}$ for
which there exists a $\Delta$-modulated functional $\Phi$ such that
$\Phi^+_{k, r}(d)(w)$ is nonincreasing in $k$, $\Phi^-_{k,r}(d)(w)$ is
nondecreasing in $k$, and $\Phi$ testifies that $\left(
\bigcap^{k-1}_{j=0} X_j \mid k \in \N \right)$ is a $\Delta$-sequence
in ${\cal F}$.
\item ${\cal F}$ is {\it closed under} $\Delta$-{\it{unions}} if
$\bigcup^\infty_{k=0} X_k \in {\cal F}$ for every union
$\Delta$-sequence $(X_k \mid k \in \N)$ in ${\cal F}$.
\item ${\cal F}$ is {\it closed under}
$\Delta$-{\it{intersections}} if $\bigcap_{k=0}^\infty X_k \in {\cal
F}$ for every intersection $\Delta$-sequence $(X_k \mid k \in \N)$ in
${\cal F}$.
\end{enumerate}
\end{definition}

\begin{lemma}\label{lutz4.18}
Let $\Delta$ be a resource bound.
Let ${\cal F}$ be a subalgebra of ${\cal F}^\nu_{R(\Delta)}$ on
$R(\Delta)$.  If $(X_k \mid k \in \N)$ is a $\Delta$-sequence in
${\cal F}$ and $\nu(X_k \mid R(\Delta)) = 0$ for all $k \in \N$, then
$(X_k \mid k \in \N)$ is a union $\Delta$-sequence in ${\cal F}$.
\end{lemma}
\begin{proof}
Assume the hypothesis, and let $\Phi$ be a witness that
$(X_k \mid k \in  \N)$ is a $\Delta$-sequence in ${\cal F}$.
Let
\[\func{\Psi}{\N \times \N \times {\cal D}_\nu}{{\cal D}_\nu \times{\cal D}_\nu}\]
be such that
\[\Psi_{k,r}(d) = \left( \sum_{j=0}^k \Phi^+_{j, j+r+1}(\mathbf{1}), d\right)\]
for all $k, r\in\N$ and $d\in\D_\nu$.
Note that the bounded sum $\sum_{j=0}^k$ is $\BFF$.
Then $\Psi$ is $\Delta$-computable since $\Delta$ is a resource bound.
Together with the proof of Lemma 4.18 in \cite{Lutz:RBM},
we have that $\Psi$ testifies that $(X_k\mid k\in\N)$ is a union $\Delta$-sequence.
\qed\end{proof}
\begin{lemma}\label{lutz4.19}
Let $\Delta$ be a resource bound.
Let ${\cal F}$ be a subalgebra of ${\cal F}^\nu_{R(\Delta)}$ on
$R(\Delta)$.  Then a sequence $(X_k \mid k \in \N)$ is a union
$\Delta$-sequence in ${\cal F}$ if and only if the complemented
sequence $(X^{\rm c}_k \mid k \in \N)$ is an intersection
$\Delta$-sequence in ${\cal F}$.  Thus ${\cal F}$ is closed under
$\Delta$-unions if and only if ${\cal F}$ is closed under
$\Delta$-intersections.
\end{lemma}
\begin{proof}
Note that for any $\nu$-splitting operator $\Phi$,
$\Phi=(\Phi^+, \Phi^-)$ is a $\Delta$-$\nu$-splitting operator
if and only if $(\Phi^-, \Phi^+)$ is a $\Delta$-$\nu$-splitting operator,
since $\Delta$ is closed under $\BFF$.
The rest of the proof is identical to that of Lemma 4.19 in \cite{Lutz:RBM}.
\qed\end{proof}

\begin{theorem}[Lutz \cite{Lutz:RBM}]\label{lutz4.20}
Let $\Delta$ be a resource bound.
\begin{enumerate}[1.]
\item ${\cal F}^\nu_{R(\Delta)}$ is closed under $\Delta$-unions
and $\Delta$-intersections.
\item If $(X_k \mid k \in \N)$ is a union $\Delta$-sequence in
${\cal F}^\nu_{R(\Delta)}$, then
\[
\nu(\cup_{k=0}^\infty X_k \mid R(\Delta)) \leq \sum_{k=0}^\infty
\nu(X_k \mid R(\Delta)),
\]
with equality if the sets $X_0, X_1, \ldots$ are disjoint.
\item If $(X_k \mid k \in \N)$ is a union $\Delta$-sequence in
${\cal F}^\nu_{R(\Delta)}$ with each $X_k \subseteq X_{k+1}$, then
\[
\nu( \cup_{k=0}^\infty X_k \mid R(\Delta)) = \lim_{k \rightarrow
\infty} \nu(X_k \mid R(\Delta)).
\]
\item If $(X_k \mid k \in \N)$ is an intersection
$\Delta$-sequence in ${\cal F}^\nu_{R(\Delta)}$ with each $X_{k+1}
\subseteq X_k$, then
\[
\nu\left(\cap_{k=0}^\infty X_k \mid R(\Delta)\right) = \lim_{k
\rightarrow \infty} \nu \left(X_k \mid R(\Delta) \right).
\]
\end{enumerate}
\end{theorem}
\begin{proof}
Let $(X_k \mid k \in \N)$ be a union $\Delta$-sequence in ${\cal
F}^\nu_{R(\Delta)}$, with the functional $\Phi$ as witness, and let $X= \bigcup_{k=0}^\infty X_k$.  Fix a functional $\Gamma$ testifying
that $\Phi$ is $\Delta$-modulated, and define a functional
\[
\func{\Theta}{\N \times {\cal D}_\nu}{{\cal D}_\nu \times {\cal D}_\nu}
\]
such that for all $r\in\N$ and $d\in\D_\nu$,
\[
\Theta(r,d) = (\Phi^+_{\infty, r+1}, \Phi^-_{m, r+1}(d)),
\]
where $m = \Gamma( r+1, r+1, d, \lambda)$.
Note that $\Theta$ is $\BFF$ in $\Phi$ and $\Gamma$.
By our choice of $\Phi$ and $\Gamma$, $\Theta$ is in $\BFF(\Delta)$
and thus $\Delta$-computable. Then $\Theta$ is a $\nu$-measurement of $X$ in $R(\Delta)$
by the proof of Theorem 4.20 in \cite{Lutz:RBM}.
Therefore $X \in \F^\nu_{R(\Delta)}$.  This, together
with Lemma \ref{lutz4.19}, proves part 1 of the theorem.

Part 2 of this theorem then follows from the proof of the part 2
of Theorem 4.20 in \cite{Lutz:RBM}.

To prove part 3 of the theorem, let $(X_k \mid k \in \N)$ be a union
$\Delta$-sequence in ${\cal F}^\nu_{R(\Delta)}$ with functional $\Phi$ as a witness
such that $X_k\subseteq X_{k+1}$  for each $k\in\N$.
Therefore, $\Phi$ also testifies that
$\left(\bigcup^{k-1}_{j=0} X_j \mid k \in \N\right)$
is a $\Delta$-sequence in $\F$.

Let $Y_0 = X_0$ and, for each $k \in \N$, let $Y_{k+1} = X_{k+1} - X_k$.
Note that since $X_k\subseteq X_{k+1}$,
\[\bigcup_{j=0}^k X_j = X_k =\bigcup_{j=0}^k Y_j.\]
Therefore $\Phi$ also testifies that
$\left(\bigcup^{k-1}_{j=0} Y_j \mid k \in \N\right)$
is a $\Delta$-sequence in $\F$ and hence testifies that $(Y_k \mid k \in \N)$
is a union $\Delta$-sequence in $\F^\nu_{R(\Delta)}$.
The rest of the proof of part 3 is identical to that of
Theorem 4.20 in \cite{Lutz:RBM}, which can be argued using
part 2 of this theorem, Theorem \ref{th:Rdeltaalgebra}, and
taking a limit over $k$ of the measures $\nu(X_k\mid R(\Delta))$.

The proof of part 4 is identical to that in the proof of Theorem 4.20
in \cite{Lutz:RBM}.  
\qed\end{proof}


%% file: pasting_lemma.tex
\begin{lemma}[Pasting Lemma]\label{lm:5_6}
\label{thm:pasting_lemma}
Let $\Delta$ be a resource bound, $p, q$ be natural numbers and $f, g,
k: \R^{p} \to \R^{q}$ be uniformly continuous, $\Delta$ computable
functions. Let $h: \R^{p} \to \R^{q}$ be the piecewise function
defined by
\begin{align*}
h(\vec{r}) &= \begin{cases} f(\vec{r}) &\text{ if } k(\vec{r}) \ge 0\\
g(\vec{r}) &\text{ otherwise.}
\end{cases}
\end{align*}
If $h$ is continuous everywhere, then it is $\Delta$-computable.
\end{lemma}
\begin{proof}
Let $\hat{f}, \hat{g}, \hat{k}$ be the respective computations of
appropriate types of $f, g,$ and $k$. Let $m$ be the maximum of the
modulus functions of $f, g$ and $k$. (For the modulus function to be
type-1, we need uniform continuity.) We define the following
functional $\hat{h}$ and prove that it is a $\Delta$-
computation of $h$.
\begin{align*}
\hat{h}(\vec{\hat{r}}, 0^{n}) &= \begin{cases}
  \hat{f}(\vec{\hat{r}})(0^{m(n)+1}) &\text{ if }
  \hat{k}(\vec{\hat{r}})(0^{m(n)+1}) \ge 0\\
\hat{g}(\vec{\hat{r}})(0^{m(n)}+1) &\text{ otherwise. }
\end{cases}
\end{align*}

If $k(\vec{r})$ and $\hat{k}(\vec{\hat{r}})(0^{m(n)+1})$ are of the same
sign, then $|h(r) - \hat{h}(\vec{\hat{r}}, 0^n)| < 2^{-n}$ by the
property of the witnesses $\hat{f}$ and $\hat{g}$.

If $k(\vec{r}) < 0$ and $\hat{k}(\vec{\hat{r}})(0^{m(n)+1}) \ge 0$, we
have that $\hat{k}(\vec{\hat{r}})(0^{m(n)}) < 2^{-(n+1)}$ by the
approximation property of $k$. We also have
$\hat{h}(\vec{\hat{r}},0^{n}) = \hat{f}(\vec{\hat{r}})(0^{m(n)+1})$, and
$h(\vec{r}) = g(\vec{r})$.
Thus,
\begin{align*}
|\hat{h}(\vec{\hat{r}}, 0^n) - h(\vec{r})| &= |\hat{h}(\vec{\hat{r}},
 0^n) - g(\vec{r})|\\
&= |\hat{f}(\vec{\hat{r}})(0^{m(n)+1}) - g(\vec{r})|.
\end{align*}
Since $k$ changes sign in the $2^{-[m(n)+1]}$ neighborhood of
$\vec{r}$, there is a point $r_1$ in it where $k(\vec{r_1}) = 0$ (This
follows from the fact that $\R^{q}$ is a connected set.). Since $h$ is
continuous at $r_1$, we can conclude $f(r_1) = g(r_1)$.

Thus,
\begin{align*}
|\hat{f}(\vec{\hat{r}}) (0^{m(n)+1}) - g(\vec{r})| &\le
|\hat{f}(\vec{\hat{r}}  (0^{m(n)+1}) - f(\vec{r_1})| +
 |f(\vec{r_1}) - g(\vec{r})|\\
&= |\hat{f}(\vec{\hat{r}}  (0^{m(n)+1}) - f(\vec{r_1})| +
 |g(\vec{r_1}) - g(\vec{r})|\\
&\le 2^{-(n+1)} + 2^{-(n+1)}.
\end{align*}

The case when $k(\vec{r}) \ge 0$ but $\hat{k}(\vec{\hat{r}})(0^{m(n)})
< 0$ is similar.
\qed\end{proof}

We can now confirm that the ``Robin Hood function'' is $\Delta$-time
computable.

\begin{lemma}\label{lm:5_7}
Let $\Delta$ be a resource bound and let $\alpha$
be a $\Delta$-computable real number. Then the function
$m_{\alpha}: \R^2 \to \R$ defined by
\begin{align*}
m_{\alpha} = \alpha s + (1-\alpha) t
\end{align*}
is a $\Delta$-computable uniformly continuous function.
\end{lemma}

\begin{lemma}\label{lm:5_8}
Let $\alpha$ be a $\Delta$-computable real number in $(0,1)$,
$H_\alpha$ be the half-plane
$$ H_\alpha = \{ (x,y) |\quad x, y \in \R, m_{\alpha}(x,y) \ge 1\} $$
and $D_\alpha = [0, \infty)^2 \cup H_\alpha$.  Then the ``Robin Hood
function'' $rh_\alpha: D_\alpha \to [0,\infty)^2$ defined by

\begin{align*}
rh_\alpha (s,t) &=  \begin{cases}
(s,t) &\text{ if } (s,t) \in [0,1]^2\\\\
(m_{\alpha}(s,t), m_{\alpha}(s,t)) &\text{ if } m_{\alpha}(s,t) \ge
  1\\\\
\left(1, \frac{m_{\alpha}(s-1,t)}{1-\alpha}\right) &\text{ if } 
 m_{\alpha}(s,t) < 1, s \ge 1, t \ge 0\\\\
\left(\frac{m_{\alpha}(s,t-1)}{\alpha}, 1\right) &\text{ if } 
 m_{\alpha}(s,t) < 1, t \ge 1, s \ge 0.
\end{cases}
\end{align*}

is $\Delta$-computable.
\end{lemma}
\begin{proof}
The Robin Hood function is a continuous piecewise linear mapping from
the Euclidean plane to itself, and we have $\alpha \in \Delta$. Hence
each component of the Robin Hood function is $\Delta$-computable. The
regions of the Robin Hood function are defined by the lines
\begin{enumerate}
\item $y = 1$.
\item $x = 1$.
\item $m_{\alpha}(s,t) - 1 = 0$.
\item $y = 0$ and $x = 0$.
\end{enumerate}

All of these are linear functions, hence all of them are
$\Delta$-computable uniformly continuous functions. Inside
$D_{\alpha}$, $0 \le s \le 1$ if and only if $(1-s)s \ge 0$. Define
$\textnormal{\scshape inside} : D_{\alpha} \to \R$ by

$$\textnormal{\scshape inside} (s,t) = \min\{(1-s)s, (1-t)t\}.$$

We conclude that $\textnormal{\scshape
inside}(s,t) \ge 0$ if and only if $(s,t) \in [0,1]^2$. Also, in
$D_{\alpha}$, a point $(s,t)$ is inside the triangle defined by the
lines $x = 1$, $y = 0$ and $m_{\alpha}(x,y) = 1$ if and only if
$(s-1)t(1-m_{\alpha}(s,t)) \ge 0$. It follows that each of the
following functions is $\Delta$-computable, by Lemma
\ref{thm:pasting_lemma}.
\begin{align*}
\\\\
rh_{\alpha}(s,t) &= \begin{cases}
 (s,t) &\text{ if } \textnormal{\scshape inside}(s,t) \ge 0\\
 rh1_{\alpha}(s,t) &\text{ otherwise. }\\
\end{cases}
\\\\
rh1_{\alpha}(s,t) &= \begin{cases}
 (m_{\alpha}(s,t),m_{\alpha}(s,t)) &\text{ if } m_{\alpha}(s,t) \ge 1\\
 rh2_{\alpha}(s,t) &\text{ otherwise. }\\
\end{cases}
\\\\
rh2_{\alpha}(s,t) &= \begin{cases}
 \left(1, \frac{m_{\alpha}(s-1,t)}{1-\alpha}\right) &\text{ if }
  (s-1)t(1-m_{\alpha}(s,t))\ge 0\\
 rh3_{\alpha} &\text{ otherwise. }\\
\end{cases}
\\\\
rh3_{\alpha}(s,t) &= \begin{cases}
\left(\frac{m_{\alpha}(s,t-1)}{\alpha}, 1\right) &\text{ if } (t-1)s(1 -
  m_{\alpha}(s,t)) \ge 0\\
(1,1) &\text{ otherwise. }
\end{cases}
\\\\
\end{align*}
Thus the Robin-Hood function is computable.
\qed\end{proof}



%% file: bib.tex
\bibliographystyle{abbrv}
\input{allbibs.tex}

%% file: allbibs.tex
\bibliography{dim,rbm,main,random,dimrelated}